\documentclass{article}
\usepackage{fullpage}
\usepackage{amsthm,amsmath,amssymb,mathtools}
\usepackage{url}
\usepackage[dvipdfmx]{graphicx}
\usepackage[dvipdfmx]{color}
\usepackage{xcolor}
\usepackage[colorlinks,citecolor=blue,bookmarks=true,linktocpage]{hyperref}
\usepackage[normalem]{ulem}
\usepackage{color}
\usepackage[numbered]{bookmark}
\usepackage[backend=biber,url=false,arxiv=abs,maxbibnames=100,isbn=false,sortcites=true]{biblatex}
\usepackage[linesnumbered,lined,ruled,noend,vlined]{algorithm2e}
\usepackage[capitalise]{cleveref}

\theoremstyle{definition}
\newtheorem{observation}{Observation}
\newtheorem{proposition}{Proposition}
\newtheorem{lemma}{Lemma}
\newtheorem{theorem}{Theorem}

\newcommand{\convc}{{\mathcal S_{\rm cvc}}}
\newcommand{\conds}{{\mathcal S_{\rm cds}}}

\title{Enumerating minimal vertex covers and dominating sets with capacity and/or connectivity constraints\thanks{
This work is partially supported
by JSPS KAKENHI Grant Numbers 
JP20H00595, 
JP20K04973, 
JP20H05964, 
JP20H05967, 
JP21K17812, 
JP21K19765, 
JP22H03549,  
JP22H00513, 
and JP23H03344 
and JST ACT-X Grant Number JPMJAX2105. 
}}

\author{
Yasuaki Kobayashi\thanks{Hokkaido University. Email:\texttt{koba@ist.hokudai.ac.jp}} \and
Kazuhiro Kurita\thanks{Nagoya University. Email:\texttt{kurita@i.nagoya-u.ac.jp}, \texttt{ono@nagoya-u.jp}}   \and
Kevin Mann\thanks{Universit\"at Trier. Email:\texttt{mann@uni-trier.de}}        \and
Yasuko Matsui\thanks{Tokai University. Email:\texttt{yasuko@tokai-u.jp}}     \and
Hirotaka Ono\footnotemark[3]}
\date{}

\begin{document}

\maketitle
\begin{abstract}
In this paper, we consider the problems of enumerating minimal vertex covers and minimal dominating sets with capacity and/or connectivity constraints. We develop polynomial-delay enumeration algorithms for these problems on bounded-degree graphs. For the case of minimal connected vertex covers, our algorithms run in polynomial delay even on the class of $d$-claw free graphs, extending the result on bounded-degree graphs, and in output quasi-polynomial time on general graphs. To complement these algorithmic results, we show that the problems of enumerating minimal connected vertex covers, minimal connected dominating sets, and minimal capacitated vertex covers in $2$-degenerated bipartite graphs are at least as hard as enumerating minimal transversals in hypergraphs.
\end{abstract}
\section{Introduction}
Enumerating minimal or maximal vertex subsets satisfying some graph properties has been widely studied for decades since it has various applications in many fields.
For example, the problem of enumerating maximal cliques in graphs is an essential task in data mining, which is intensively studied from both theoretical and practical perspectives~\cite{conf/swat/MakinoU04,FoxRSWW:SICOMP:Finding:2020,ConteGMV:Algorithmica:Sublinear:2020,ConteVMPT:EDBT:Finding:2016,BronK:CACM:Finding:1973,TomitaTT:TCS:worst-case:2006,TsukiyamaIAS:SICOMP:New:1977}.
This enumeration problem is equivalent to that of enumerating minimal vertex covers as there is a one-to-one correspondence between the collection of all maximal cliques in $G$ and that of all minimal vertex covers of its complement $\overline{G}$.

Another important enumeration problem is to enumerate minimal dominating sets in graphs.
This enumeration problem is known to be ``equivalent''\footnote{This means that there is an output-polynomial time algorithm for enumerating minimal dominating sets in graphs if and only if there is an output-polynomial time algorithm for dualizing monotone Boolean functions.} to dualizing monotone Boolean functions or enumerating minimal transversals in hypergraphs~\cite{KanteLMN:SIDMA:Enumeration:2014}, which has many applications in a broad range of fields in computer science~\cite{BiochI95:iandc,EiterG95:siamcomp,EiterG:JELIA:Hypergraph:2002}. (See~\cite{EiterMG:DAM:Computational:2008} for a survey).

As for the complexity of these enumeration problems, the current status of the problems of enumerating maximal cliques and enumerating minimal dominating sets of graphs is significantly different.
By a classical result of Tsukiyama et al.~\cite{TsukiyamaIAS:SICOMP:New:1977}, we can enumerate all maximal cliques (and hence minimal vertex covers) of a graph in polynomial delay.
Here, an enumeration algorithm runs in \emph{polynomial delay} if the time elapsed\footnote{This also includes the running time of pre-processing and post-processing.} between every pair of two consecutive outputs is upper bounded by a polynomial solely in the input size.
In contrast to this, no output-polynomial time algorithm for enumerating minimal dominating sets of a graph is known so far, where an enumeration algorithm runs in \emph{output-polynomial time} if the total running time is upper bounded by a polynomial in the combined size of the input and all outputs\footnote{Note that the output size can be exponential in the input size.}. 
There are several results showing polynomial-delay or output-polynomial time algorithms for enumerating minimal dominating sets in several special classes of graphs~\cite{Kante:WG:Polynomial:2015,Kante:WADS:Polnomial:2015,Golovach:ICALP:Incremental:2013,KobayashiKW:arXiv:Efficient:2020,Courcelle:DAM:Linear:2009,BonamyDHPR:TALG:Enumerating:2020}.
The fastest known algorithm for general graphs is due to Fredman and Khachiyan~\cite{FredmanK:JAL:Complexity:1996}, which runs in time $N^{o(\log N)}$, where $N$ is the number of vertices and hyperedges plus the number of minimal dominating sets of an input graph.

As variants of these two enumeration problems, we study the problems of enumerating minimal \emph{connected} vertex covers and minimal \emph{connected} dominating sets of graphs, which we call \textsc{Minimal Connected Vertex Cover Enumeration} and \textsc{Minimal Connected Dominating Set Enumeration}, respectively.
Although these two combinatorial objects are considered to be natural (as they are well studied in several areas~\cite{UenoKG:DM:nonseparating:1988,Cygan:SWAT:Deterministic:2012,FominGK:Algorithmica:Solving:2008,EscoffierGM:JDA:Complexity:2010,BalakrishnanRR:WADS:Connected:1993,GuhaK:Algorithmica:Approximation:1998}, including input-sensitive enumeration algorithms~\cite{Abu-KhzamF0LM:ESA:Enumerating:2022,GolovachHK:EJC:Enumeration:18}), enumeration algorithms (with an output-sensitive analysis~\cite{FernauGS18}) for these variants are not much investigated in the literature.
Kant\'e et al.~\cite{KanteLMN:SIDMA:Enumeration:2014} showed that the problem of enumerating minimal connected dominating sets in split graphs is at least as hard as that of enumerating minimal dominating sets in general graphs, while the problem without the connectivity requirement admits a polynomial-delay algorithm on this class of graphs.
This indicates that the connected variant is also a challenging problem.
Very recently, Kobayashi et al.~\cite{KobayashiKW:arXiv:Polynomial:2024} devised a polynomial-delay algorithm for enumerating minimal connected vertex covers with cardinality at most $t$ for a given subcubic graph and a threshold $t$ by exploiting a well-known relation between a connected vertex cover in a (sub)cubic graph and a ``matching'' in a certain matroid~\cite{UenoKG:DM:nonseparating:1988}.

We also study another type of variants: the problems of enumerating minimal \emph{capacitated} vertex covers and minimal \emph{capacitated} dominating sets of graphs, which we call \textsc{Minimal Capacitated Vertex Cover Enumeration} and \textsc{Minimal Capacitated Dominating Set Enumeration}, respectively.
These two problems generalize the conventional minimal vertex cover and dominating set enumeration problems (see \Cref{sec:preli} for details).

In this paper, we tackle these enumeration problems by mainly restricting our focus to \emph{bounded-degree graphs}.
We design polynomial-delay algorithms for these problems on graphs with maximum degree $\Delta = O(1)$ in \Cref{sec:convc,sec:conds,sec:capacity}.
The polynomial-delay algorithm for \textsc{Minimal Connected Vertex Cover Enumeration} can be extended to that on $d$-claw free graphs with $d = O(1)$ (\Cref{thm:convc:claw}).
For \textsc{Minimal Connected Vertex Cover Enumeration} and \textsc{Minimal Capacitated Vertex Cover Enumeration}, we show, in \Cref{sec:convc,sec:capacity}, that these problems are at least as hard as enumerating minimal transversals in hypergraphs, even if the input graph is restricted to bipartite and $2$-degenerate.
In \Cref{sec:conds}, we also give a similar ``hardness'' result of \textsc{Minimal Connected Dominating Set Enumeration} on $2$-degenerate bipartite graphs.
These results indicate that the connectivity and/or capacity requirements make the problems ``harder'' in a certain sense.
Finally, in \Cref{sec:convc}, we give an \emph{output quasi-polynomial time} algorithm for \textsc{Minimal Connected Vertex Cover Enumeration} on general graphs.

\section{Preliminaries}\label{sec:preli}
Throughout this paper, we only consider undirected simple graphs (unless otherwise stated).
Let $G$ be a graph.
We use $n$ to denote the number of vertices in $G$.
We denote by $V(G)$ and $E(G)$ the sets of vertices and edges in $G$, respectively.
For $v \in V(G)$, the set of neighbors of $v$ in $G$ is denoted by $N_G(v)$.
This notation is extended to vertex sets: $N_G(X) = \bigcup_{v \in X}N_G(v) \setminus X$ for $X \subseteq V(G)$.
For $X \subseteq V(G)$, the subgraph of $G$ induced by $X$ is denoted by $G[X]$.
We define $N_G[v] = N_G(v) \cup \{v\}$ and $N_G[X] = N_G(X) \cup X$ for $v \in V(G)$ and $X \subseteq V(G)$.
When the subscripts are clear from the context, we may omit them.
For $v \in V(G)$, we write $G - v$ to denote the graph obtained from $G$ by deleting $v$ (and its incident edges).

A set of vertices $C$ is called a \emph{vertex cover} of $G$ if for every edge in $G$, at least one end vertex of it belongs to $C$.
A vertex cover $C$ is said to be \emph{connected} if $G[C]$ is connected.
A \emph{dominating set} of $G$ is a set of vertices $D$ such that $V(G) = N[D]$, that is, every vertex in $V(G) \setminus D$ has a neighbor in $D$.
A connected dominating set of $G$ is defined analogously.

The ``capacitated'' variants of a vertex cover and dominating set are defined as follows.
Let $c \colon V(G) \to \mathbb N$ be a capacity function.
Notice that we allow a vertex $v$ such that $c(v) = 0$.
A \emph{capacitated vertex cover} of $(G, c)$ is a pair $(C, \alpha)$ of a vertex set $C \subseteq V(G)$ and a function $\alpha \colon E(G) \to C$ such that $\alpha(\{u, v\})$ is either $u$ or $v$ and $|\alpha^{-1}(v)| \le c(v)$ for $v \in C$, where $\alpha^{-1}(v) = \{e \in E(G) : \alpha(e) = v\}$ is the set of edges mapped to $v$ under $\alpha$.
In other words, a vertex $v \in C$ covers at most $c(v)$ edges.
A \emph{capacitated dominating set} of $(G, c)$ is a pair $(D, \beta)$ of a vertex set $D$ and a function $\beta \colon V(G) \setminus D \to D$ such that $\beta(v)$ is a neighbor of $v$ and $|\beta^{-1}(v)| \le c(v)$ for $v \in D$, where $\beta^{-1}(v) = \{w \in V(G) \setminus D : \beta(w) = v\}$.
We simply refer to a vertex set $X \subseteq V(G)$ as a capacitated vertex cover (resp. capacitated dominating set) of $(G, c)$ if there is a function $\alpha \colon E(G) \to X$ (resp. $\beta \colon V(G) \setminus X \to X$) such that $(X, \alpha)$ is a capacitated vertex cover (resp. $(X, \beta)$ is a capacitated dominating set) of $(G, c)$.
Clearly, every capacitated vertex cover (resp. capacitated dominating set) of $(G, c)$ is a vertex cover (resp. dominating set) of $G$, and the converse holds when setting $c(v) = \Delta$ for $v \in V(G)$, where $\Delta$ is the maximum degree of a vertex in $G$.

Let $\mathcal S \subseteq 2^{V(G)}$ be a collection of vertex sets of $G$.
We say that $\mathcal S$ is \emph{monotone} if for $X, Y \subseteq V(G)$ with $X \subseteq Y$, $X \in \mathcal S$ implies $Y \in \mathcal S$.
It is easy to see that the collections of capacitated vertex covers and capacitated dominating sets of $(G, c)$ are monotone, while the collection of connected vertex sets (i.e., vertex sets that induce connected subgraphs of $G$) is not monotone.
The following proposition shows that the collections of connected vertex covers and connected dominating sets of $G$ are also monotone.

\begin{proposition}\label{prop:monotone}
    Let $G = (V, E)$ be a connected graph and let $X$ be a connected vertex cover (resp. connected dominating set) of $G$.
    Then, for any $Y \subseteq V(G)$ with $X \subseteq Y$, $Y$ is a connected vertex cover (resp. connected dominating set) of $G$ as well.
\end{proposition}
\begin{proof}
    Let $X$ be a connected vertex cover of $G$.
    We assume that $G$ has at least two vertices, as otherwise the lemma is trivial.
    Moreover, we assume that $X$ has at least one vertex.
    As $G$ is connected and $X$ is a vertex cover of $G$, every vertex $v$ in $V(G) \setminus X$ has a neighbor in $X$.
    This implies that $G[X \cup \{v\}]$ is connected.
    The case of a connected dominating set is analogous.
\end{proof}

We also show that one can decide in polynomial time whether a given vertex set $X$ is a capacitated vertex cover (and a capacitated dominating set) of $(G, c)$.
\begin{proposition}\label{prop:feasibility-check}
    There are polynomial-time algorithms for checking whether a given vertex set $X$ is a capacitated vertex cover and is a capacitated dominating set of $(G, c)$, respectively.
\end{proposition}
\begin{proof}
    We first consider the case for capacitated vertex cover.
    We reduce this feasibility-checking problem to the bipartite matching problem as follows.
    The bipartite graph $H$ consists of two independent sets $V_E$ and $V$.
    The first set $V_E$ is defined as $V_E = \{w_e : e \in E(G)\}$, that is, $V_E$ contains a vertex $w_e$ for each edge $e$ in $G$.
    The second set $V$ is defined as $V = \{w^i_v : v \in V(G), 1 \le i \le c(v)\}$, that is, $V$ contains $c(v)$ vertices for each vertex $v$ in $V(G)$.
    For $v \in V(G)$, $1 \le i \le c(v)$, and $e \in E(G)$, we add an edge between $w_e$ and $w^i_v$ if $e$ is incident to $v$ in $G$.
    The graph constructed in this way is indeed bipartite, which we denote by $H$.
    
    It is easy to observe that $(G, c)$ has a capacitated vertex cover $(X, \alpha)$ if and only if $H[V_E \cup \{w^i_v : v \in X, 1 \le i \le c(v)\}]$ has a matching $M$ saturating $V_E$, as the function $\alpha$ is straightforwardly defined from the matching 
    and vice versa.

    For capacitated dominating set, we modify the construction of $H$ as follows: Replace $V_E$ with $V' = \{w_v : v \in V\}$ and add edges between $w_{u} \in V'$ and $w^i_v \in V$ for all $1 \le i \le c(v)$ if and only if $\{u, v\} \in E(G)$.
    Similarly, $(G, c)$ has a capacitated dominating set $(X, \beta)$ if and only if $H[\{w_v : v \in V(G) \setminus X\} \cup \{w^i_v : v \in X, 1 \le i \le c(v)\}]$ has a matching saturating $\{w_v : v \in V(G) \setminus X\}$.

    We can find a maximum cardinality bipartite matching in polynomial time, proving this proposition.
\end{proof}

For a non-negative integer $k$, a vertex ordering $(v_1, \dots, v_n)$ of $G$ is called a \emph{$k$-degenerate ordering} of $G$ if each vertex $v_i$ has at most $k$ neighbors in $G[\{v_i, \dots, v_n\}]$.
A graph is said to be \emph{$k$-degenerate} if it admits a $k$-degenerate ordering~\cite{Lick:CJM:1970}. 

\section{A quick tour of the supergraph technique}
Before proceeding to our algorithms, we quickly review the \emph{supergraph technique} (also known as \emph{$X - e + Y$ method}), which is frequently used in designing enumeration algorithms~\cite{Boros:COCOON:Generating:2007,Khachiyan:ESA:Enumerating:2006,ConteGMUV:SICOMP:Proximity:2022,CohenKS:JCSS:Generating:2008,SchwikowskiS:DAM:enumerating:2002}.
The crux is summarized below in \Cref{thm:supergraph}.

Let $\mathcal S \subseteq 2^U$ be a collection of subsets of a finite set $U$.
Here, we assume that every pair of distinct sets in $\mathcal S$ is incomparable with respect to set inclusion, that is, for $X, Y \in \mathcal S$, both $X \setminus Y$ and $Y \setminus X$ are nonempty unless $X = Y$.
The basic idea to the supergraph technique is to define a strongly connected directed graph $\mathcal D = (\mathcal S, \mathcal A)$ on $\mathcal S$ by appropriately defining an arc set $\mathcal A$.
Given this directed graph, we can enumerate all sets in $\mathcal S$ by solely traversing $\mathcal D$ from an arbitrary $S \in \mathcal S$.
To this end, we need to define the directed graph $\mathcal D$ so that it is strongly connected.

We first observe that, under the above incomparability on $\mathcal S$, for $X, Y \in \mathcal S$, $|X \setminus Y| = 0$ if and only if $X = Y$.
We define a set of arcs of $\mathcal{D}$ in such a way that for any two $X, Y \in \mathcal{S}$ with $X \neq Y$, $X$ has an outgoing arc to some $Z \in \mathcal{S}$ such that $|X \setminus Y| > |Z \setminus Y|$.
The (out-)neighborhood of $X$ defined in this way is denoted by $\mathcal N^+(X)$.
By the above observation, this directed graph $\mathcal D$ is strongly connected, as we can inductively show that there is a directed path from $X$ to $Y$ in $\mathcal D$ via $Z \in \mathcal N^+(X)$.
This idea is formalized as follows.
\begin{theorem}[e.g., \cite{KobayashiKW:arXiv:Efficient:2020,CohenKS:JCSS:Generating:2008,ConteGMUV:SICOMP:Proximity:2022}]\label{thm:supergraph}
    Suppose that for $X, Y \in \mathcal S$ with $X \neq Y$, the neighborhood $\mathcal N^+(X)$ of $X$ contains a set $Z \in \mathcal S$ such that $|X \setminus Y| > |Z \setminus Y|$.
    Moreover, suppose that given $X \in \mathcal S$, we can compute the neighborhood of $X$ in (total) time $T(n)$, where $n = |U|$.
    Then, given an arbitrary initial set $S \in \mathcal S$, we can enumerate all sets in $\mathcal S$ in delay $T(n)\cdot n^{O(1)}$.
\end{theorem}

Due to \Cref{thm:supergraph}, it suffices to define such a polynomial-time computable neighborhood $\mathcal N^+\colon \mathcal S \to 2^{\mathcal S}$.

\section{\textsc{Minimal Connected Vertex Cover Enumeration}}\label{sec:convc}

\subsection{Bounded-degree graphs}
In order to enumerate all minimal connected vertex covers in a graph, it suffices to construct a directed graph discussed in the previous section.
To be more precise, let $G$ be a connected graph with maximum degree $\Delta$ and let $\convc$ be the collection of all minimal connected vertex covers of $G$.
Let $X, Y \in \convc$ with $X \neq Y$.
By the minimality of $X$ and $Y$, there is a vertex $v \in X \setminus Y$.
As $Y$ is a vertex cover of $G$ not including $v$, all the neighbors of $v$ are included in $Y$ (i.e., $N(v) \subseteq Y$).
We define a vertex cover $X'$ of $G$ as $X' \coloneqq (X \setminus \{v\}) \cup N(v)$.
By the connectivity of $G[X]$, we have the following observation.

\begin{observation}
    Each component of $G[X']$ has at least one vertex of $N(v)$.
\end{observation}

This observation indicates that the number of connected components in $G[X']$ is upper bounded by $\Delta$.

\begin{lemma}\label{lem:convc:augment}
    Let $S \subseteq V(G)$ be an arbitrary vertex cover of $G$ and let $Y$ be a minimal connected vertex cover of $G$.
    Suppose that each component in $G[S]$ has at least one vertex of $Y$.
    Then, there is a vertex set $W \subseteq Y$ of at most $q-1$ vertices such that $G[S \cup W]$ is connected, where $q$ is the number of connected components in $G[S]$.
\end{lemma}
\begin{proof}
    We prove the lemma by induction on the number of connected components $q$ in $G[S]$.
    If $G[S]$ is connected, we are done.
    Suppose otherwise.
    Let $C$ and $C'$ be distinct components in $G[S]$.
    Since both $C \cap Y$ and $C' \cap Y$ are nonempty, there is a path $P$ in $G[Y]$ connecting $C \cap Y$ and $C' \cap Y$.
    We choose $C$ and $C'$ to minimize the length of such a path $P$.
    Since $S$ is a vertex cover of $G$, the length of $P$ is exactly $2$.
    Let $w$ be the (unique) internal vertex of $P$.
    Applying the induction to $S \cup \{w\}$ proves the lemma.
\end{proof}

We say that a vertex set $W$ is a \emph{valid augmentation} for $X'$ if $G[X' \cup W]$ is connected.
In particular, a valid augmentation $W$ for $X'$ is called a \emph{$Y$-valid augmentation} for $X'$ if $W \subseteq Y$.
Thus, \Cref{lem:convc:augment} ensures that there is a $Y$-valid augmentation $W$ for $X'$ of size at most $\Delta - 1$ for any $Y \in \convc$.
As $v \in X \setminus Y$, $N(v) \subseteq Y$, and $W \subseteq Y$, we have
\begin{align*}
    |X \setminus Y| &\ge |(X' \cup \{v\}) \setminus (N(v) \cup Y)|
    > |X' \setminus Y| 
    = |(X' \cup W) \setminus Y| 
    \ge |Z \setminus Y|,
\end{align*}
where $Z$ is an arbitrary minimal connected vertex cover of $G$ with $Z \subseteq X' \cup W$.

Now, we formally define $\mathcal N^+(X)$ for $X \in \convc$.
For a connected vertex cover $C$ of $G$, we let $\mu(C)$ be an arbitrary minimal connected vertex cover of $G$ with $\mu(C) \subseteq C$.
By~\Cref{prop:monotone}, $\mu(C)$ can be computed in polynomial time by greedily removing vertices from $C$.
We define
\begin{align*}
    \mathcal N^+(X) &= \{\mu((X \setminus \{v\}) \cup N(v) \cup W) :\\
    &\hspace{1cm} v \in X, W \text{ is a valid augmentation for } (X \setminus \{v\}) \cup N(v), |W| \le \Delta - 1\},
\end{align*}
which can be computed in $n^{\Delta + O(1)}$ time by just enumerating all valid augmentations $W$ of size at most $\Delta - 1$.
As discussed above, for each $Y \in \convc$ with $X \neq Y$, $\mathcal N^+(X)$ has a minimal connected vertex cover $Z \in \convc$ such that $|X \setminus Y| > |Z \setminus Y|$.
By~\Cref{thm:supergraph}, we have the following.

\begin{theorem}\label{thm:convc:degree}
    There is an $n^{\Delta + O(1)}$-delay enumeration algorithm for \textsc{Minimal Connected Vertex Cover Enumeration}.
\end{theorem}

\paragraph{Extension to $d$-claw free graphs.}
The running time of the algorithm in \Cref{thm:convc:degree} depends, in fact, on the number of connected components in $G[X']$ with $X' = (X \setminus \{v\}) \cup N(v)$ rather than the degree of $v$.
We can extend this idea as follows.
For $d \in \mathbb N$, a \emph{$d$-claw} is a star graph with $d$ leaves.
A graph $G$ is said to be \emph{$d$-claw free} if it has no $d$-claw as an induced subgraph.
Clearly, every graph with maximum degree $\Delta$ is $(\Delta + 1)$-claw free.
Suppose that $G$ is $d$-claw free.
As the class of $d$-claw free graphs is hereditary, every induced subgraph of $G$ is also $d$-claw free.
A crucial observation to our extension is that for every vertex $v \in X$, $G[(X \setminus \{v\}) \cup N(v)]$ has at most $d - 1$ components as otherwise $G$ has an induced $d$-claw with center~$v$.
This observation shows that there always exists a valid augmentation for $(X \setminus \{v\}) \cup N(v)$ of size~$d - 2$, which yields a polynomial-delay algorithm for $d = O(1)$.

\begin{theorem}\label{thm:convc:claw}
    There is an $n^{d + O(1)}$-delay enumeration algorithm for \textsc{Minimal Connected Vertex Cover Enumeration}, provided that the input graph $G$ is $d$-claw free.
\end{theorem}

\subsection{General graphs}
We complement the polynomial-delay enumeration algorithm in the previous subsection by showing that \textsc{Minimal Connected Vertex Cover Enumeration} is ``not so easy'', even on bipartite graphs.
For a hypergraph $\mathcal H = (V, \mathcal E)$, a \emph{transversal} of $\mathcal H$ is a vertex subset $S \subseteq V$ such that each hyperedge contains at least one vertex of~$S$.
Recall that no output-polynomial time algorithm for enumerating minimal transversals in hypergraphs is known so far, and several similar ``hardness'' results are discussed in the literature~\cite{ConteGKMUW19,BergougnouxDI23}.

\begin{theorem}\label{thm:cvc:mds-hard}
    If there is an output-polynomial time (or a polynomial-delay) algorithm for \textsc{Minimal Connected Vertex Cover Enumeration} on $2$-degenerate bipartite graphs,
    then there is an output-polynomial time (or a polynomial-delay) algorithm for enumerating minimal transversals in hypergraphs.
\end{theorem}
\begin{proof}
    We first give a construction of a bipartite graph $G$ from a hypergraph $\mathcal H$ and show that there is a bijection between the collection of minimal connected vertex covers of $G$ and that of minimal transversals of~$\mathcal H$.
    We then transform the bipartite graph $G$ into a $2$-degenerate bipartite graph while keeping this bijection.
    
    Let $\mathcal H = (V, \mathcal E)$ be a hypergraph.
    We construct a bipartite graph $G$ as follows.
    We start with the incidence bipartite graph of $\mathcal H$, that is, the vertex set of $G$ consists of two independent sets $V$ and $V_{\mathcal E} = \{w_e : e \in \mathcal E\}$, and two vertices $v \in V$ and $w_e \in V_{\mathcal E}$ are adjacent if and only if $v \in e$.
    From this bipartite graph, we add a pendant vertex $w_e'$ that is only adjacent to $w_e$ for each $w_e \in V_{\mathcal E}$ and a vertex $r$ that is adjacent to all vertices in $V$.
    To complete the construction of $G$, we add a pendant vertex $r'$ that is adjacent to $r$.
    See \Cref{fig:convc:reduction}~(a) for an illustration.
    \begin{figure}[t]
        \centering
        \includegraphics{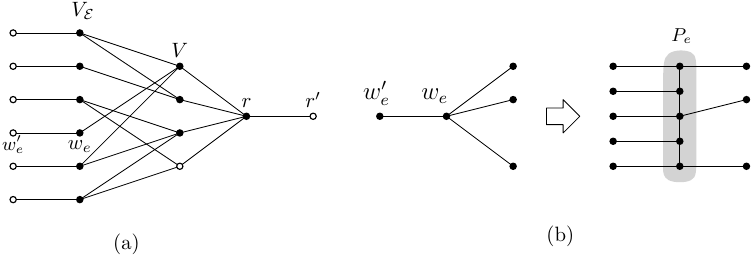}
        \caption{The left figure illustrates the constructed graph $G$. Filled circles indicate a minimal connected vertex cover of $G$. The right figure illustrates a part of the construction of a $2$-degenerate graph $G'$ from $G$.}
        \label{fig:convc:reduction}
    \end{figure}
    In the following, we show that there is a bijection between the collection of minimal connected vertex covers of $G$ and that of minimal transversals of $\mathcal H$.

    Let $C$ be an arbitrary minimal connected vertex cover of $G$.
    We first observe that $w_e \in C$ and $w_e' \notin C$ for all $e \in \mathcal E$.
    To see this, fix a hyperedge $e \in \mathcal E$.
    As $C$ is a vertex cover of $G$, at least one of $w_e$ and $w'_e$ is contained in $C$.
    If both of them are contained in $C$, $C \setminus \{w_e'\}$ is still a connected vertex cover of $G$, contradicting to the minimality.
    Thus, exactly one of $w_e$ and $w_e'$ is contained in $C$.
    As $G[C]$ is connected, $w_e$ must be contained in $C$.
    Similarly, we have $r \in C$ and $r' \notin C$.

    Let $S = C \cap V$.
    It is easy to see that $C \setminus S = V_{\mathcal E} \cup \{r\}$ is a vertex cover of~$G$.
    Due to the connectivity of $C$, for every $e \in \mathcal E$, $C$ contains at least one vertex of $V$ that is adjacent to $w_e$.
    This implies that $S$ is a transversal of $\mathcal H$.
    If there is $v \in S$ such that $S \setminus \{v\}$ is a transversal of $\mathcal H$, $C \setminus \{v\}$ is a connected vertex cover of $G$ as every $w_e \in \mathcal V_{\mathcal E}$ has a neighbor in $S \setminus \{v\}$, contradicting the minimality of~$C$.
    Thus, $S$ is a minimal transversal of $\mathcal H$.
    It is straightforward to verify that this relation is reversible, that is, for every minimal transversal $S$ of $\mathcal H$, $S \cup V_{\mathcal E} \cup \{r\}$ is a minimal connected vertex cover of $G$.

    From a hypergraph $\mathcal H$, we can construct the graph $G$ in polynomial time.
    Therefore, if there is an output-polynomial time algorithm for \textsc{Minimal Connected Vertex Cover Enumeration}, we can enumerate all minimal transversals of $\mathcal H$ in output-polynomial time as well.

    Now, we modify a bipartite graph $G$ into a $2$-degenerate bipartite graph $G'$.
    Let $w_e$ be a vertex in $V_{\mathcal E}$ and $\{u_1, \ldots, u_{d}\} \subseteq V$ be the set of vertices that are contained in $e$.
    We replace $w_e \in V_{\mathcal E}$ and the pendant vertex $w'_e$ with a path $P_e = (p^e_1, \ldots, p^e_{2d+1})$ and $2d + 1$ pendant vertices, each of which is adjacent to $p^e_i$ for each $p^e_i$.
    Furthermore, we add an edge $\{u_j, p^e_{2j-1}\}$ for each $1 \le j \le d$.
    \Cref{fig:convc:reduction}~(b) illustrates this replacement.
    We do this for each $e \in \mathcal E$, and it is easy to observe that the graph $G'$ obtained in this way is bipartite.
    Moreover, we claim that $G$ is $2$-degenerate: we first eliminate all the pendant vertices of $G'$, the vertices of $P_e$ that have even indices, the remaining vertices of $P_e$, the vertices in $V$, and then $r$.
    This elimination ordering is indeed a $2$-degeneracy ordering, that is, each eliminated vertex has at most two neighbors in the graph induced by the non-eliminated vertices. 
    The pendant vertices attached to $P_e$ ensures that all the vertices of $P_e$ are included in every minimal connected vertex cover of $G'$, and hence there is a one-to-one correspondence between the set of minimal connected vertex covers of $G'$ and the set of minimal transversals of $\mathcal H$ as well.
\end{proof}

We next discuss an extension of the previous algorithm to general graphs.
Although the algorithm described in the previous subsection also works for general graphs, its delay can be exponential in $n$.
In order to cope with this exponential running time bound, we refine the definition of $\mathcal N^+$ and prove that $\mathcal N^+(X)$ can be generated with a better running time bound for $X \in \convc$.

We first observe that the algorithm also works even when $\mathcal N^+$ is modified into
\begin{align*}
    \mathcal N^+(X) = \{\mu((X \setminus \{v\}) \cup N(v) \cup W) : v \in X, W \text{ is a minimal valid augmentation for } (X \setminus \{v\}) \cup N(v)\},
\end{align*}
where a \emph{minimal valid augmentation} for a vertex cover $X'$ is an inclusion-wise minimal $W$ such that $G[X' \cup W]$ is connected.
In other words, the directed graph $\mathcal D$ defined on this modified neighborhood is strongly connected as well.
In the following, we show that for $X \in \convc$, $\mathcal N^+(X)$ can be generated in time $(n + |\mathcal N^+(X)|)^{O(\log n)}$, which yields an \emph{output quasi-polynomial time bound} $(n + |\convc|)^{O(\log n)}$ for \textsc{Minimal Connected Vertex Cover Enumeration} due to \Cref{thm:supergraph}.

In order to enumerate all the sets in $\mathcal N^+(X)$, it suffices to enumerate all minimal valid augmentations for $(X \setminus \{v\}) \cup N(v)$ for each $v \in X$.
Let $v \in X$ and let $X' = (X \setminus \{v\}) \cup N(v)$.
Recall that $X'$ is a vertex cover of $G$.
To simplify the following discussion, we construct a bipartite graph $H$ with color classes $L$ and $R$ as follows.
For each connected component $C$ in $G[X']$, $L$ contains a vertex $v_C$, and define $R \coloneqq V(G) \setminus X'$.
Note that $V(G) \setminus X'$ is an independent set of $G$.
For $v_C \in L$ and $w \in R$, they are adjacent in $H$ if and only if $w$ has a neighbor in $C$.
In other words, the bipartite graph $H$ is obtained from $G$ by contracting each connected component $C$ in $G[X']$ into a single vertex $v_C$.
Then, every (minimal) valid augmentation for $X'$ is a vertex set $W \subseteq R$ such that $H[L \cup W]$ is connected, and vice versa.
Moreover, such a set $W$ is also a (minimal) valid augmentation for $L$ (in $H$).
Hence, our subsequent discussion will be done on the bipartite graph $H$.
In the remaining of this section, we assume that $H$ is connected.

\begin{observation}\label{obs:minimal-augmentation-cutver}
    Let $W \subseteq R$ be a valid augmentation for $L$ in $H$.
    Then, $W$ is minimal if and only if each vertex of $W$ is a cut vertex in $H[L \cup W]$.
\end{observation}

We introduce some notations to explain our algorithm.
The collection of minimal valid augmentations in $H$ for $L$ is denoted by $\mathcal S(H, L)$.
For a collection of valid augmentations $\mathcal A$ in $H$, 
the set of minimal valid augmentations in $\mathcal A$ is denoted by ${\tt min}_{H,L}(\mathcal A)$ (i.e., ${\tt min}_{H,L}(\mathcal A) = \mathcal A \cap \mathcal S(H, L)$).
For a collection of sets $\mathcal U = \{U_1, \dots, U_t\}$ and an element $e$, $\mathcal U + e$ denotes
$\{U_1 \cup \{e\}, \ldots, U_t \cup \{e\}\}$.
We extend this notation to a set of elements: For a set $X$, $\mathcal U + X$ denotes $\{U_1 \cup X, \ldots, U_t \cup X\}$.
For $X \subseteq R$, we define a bipartite graph $H_X$ obtained from $H$ by contracting each connected component $C$ in $H[L \cup X]$ into a single vertex $v_C$.
We also define
\begin{align*}
    L_X = \{v_C : C \text{ is a component in } H[L \cup X] \text{ containing a vertex in } L\}.
\end{align*}
Clearly, we have $H_\emptyset = H$ and $L_\emptyset = L$.
When $X = \{v\}$, we respectively denote $H_{\{v\}}$ and $L_{\{v\}}$ by $H_v$ and $L_v$ as shorthand notations.

\begin{lemma}\label{lem:cvc:contraction-minimality}
    Let $W \in \mathcal S(H, L)$.
    For $W' \subseteq W$, $W \setminus W'$ is a minimal valid augmentation for $L_{W'}$ in $H_{W'}$.
\end{lemma}
\begin{proof}
    Since $H_{W'}[L_{W'} \cup (W \setminus W')]$ is obtained from $H[L \cup W]$ by contracting each component $C$ of $H[L \cup W]$ containing a vertex in $W'$ into a single vertex $v_C$, $H_{W'}[L_{W'} \cup (W \setminus W')]$ is connected, and hence $W \setminus W'$ is a valid augmentation for $L_{W'}$ in $H_{W'}$.
    Since $W \in \mathcal S(H, L)$, by~\Cref{obs:minimal-augmentation-cutver}, each vertex in $W$ is a cut vertex in $H[L \cup W]$.
    Moreover, as $W$ is an independent set in $H$, all the neighbors of $W \setminus W'$ are preserved in $H_{W'}$.
    This implies that every vertex in $W \setminus W'$ is still a cut vertex in $H_{W'}[L_{W'} \cup (W \setminus W')]$.
    Hence, by~\Cref{obs:minimal-augmentation-cutver}, $W \setminus W'$ is a minimal valid augmentation for $L_{W'}$ in $H_{W'}$.
\end{proof}

\begin{lemma}\label{lem:cvc:contraction-extension}
    Let $X \subseteq R$.
    For $W \in \mathcal S(H_X, L_X)$, $W \cup X$ is a valid augmentation for $L$ in $H$.
\end{lemma}
\begin{proof}
    We claim that there is a path between $u$ and $v$ in $H[L \cup W \cup X]$ for $u, v \in L$, which implies that $H[L \cup W \cup X]$ is connected, as each vertex in $W \cup X$ has a neighbor in $H$.
    If $u$ and $v$ are contained in a single component in $H[L \cup X]$, we are done.
    Thus, we assume that $u$ and $v$ belong to distinct components $C$ and $C'$, respectively.
    Since $H_X[L_X \cup W]$ is connected, there is a path $P_X$ between $v_C$ and $v_{C'}$ in $H_X[L_X \cup W]$.
    From this path $P_X$, we can construct a path between $u$ and $v$ in $H[L \cup W \cup X]$ by appropriately replacing a vertex $v_{C''}$ in $P_X$ with a path in component $C''$.
    Hence, $W \cup X$ is a valid augmentation for $L$ in $H$.
\end{proof}

Our algorithm shares the same underlying idea as \cite{FredmanK:JAL:Complexity:1996,Tamaki00}, which is described in \Cref{alg:oq}.
The algorithm recursively generates minimal valid augmentations for $L$ in $H$ unless $\min(|L|, |R|) \le 1$.
Let $v$ be an arbitrary vertex in $R$.
We enumerate all sets in $\mathcal S(H, L)$ by separately generating those including $v$ and those excluding~$v$.
The following lemma ensures that this strategy correctly enumerates all minimal valid augmentations.

\begin{lemma}\label{lem:large}
    $\mathcal S(H, L) = \mathcal S(H - v, L) \cup {\tt min}_{H, L}(\mathcal S(H_v, L_v) + v)$.
\end{lemma}
\begin{proof}
    By~\Cref{lem:cvc:contraction-extension}, $\mathcal S(H_v, L_v) + v$ is a valid augmentation for $L$ in $H$.
    Since each valid augmentation in $\mathcal S(H - v, L)$ is also a minimal valid augmentation for $L$ in $H$, we have
    \begin{align*}
        \mathcal S(H - v, L) \cup {\tt min}_{H, L}(\mathcal S(H_v, L_v) + v) \subseteq \mathcal S(H, L).
    \end{align*}

    Let $W \in \mathcal S(H, L)$ be a minimal valid augmentation for $L$ in $H$.
    If $v$ is not contained in $W$, $W$ is a minimal valid augmentation for $L$ in $H - v$ as well, implying that $W \in \mathcal S(H - v, L)$.
    Suppose otherwise.
    By~\Cref{lem:cvc:contraction-minimality}, it holds that $W \setminus \{v\} \in \mathcal S(H_v, L_v)$.
    Hence, we have $W \in {\tt min}_{H, L}(\mathcal S(H_v, L_v) + v)$.
\end{proof}

By~\Cref{lem:large}, we can enumerate all minimal valid augmentations in $\mathcal S(H, L)$ by enumerating those in $\mathcal S(H - v, L)$ and in $\mathcal S(H_v, L_v)$ separately and just combining them.
The following lemma gives another way to enumerate them.

\begin{lemma}\label{lem:small}
    $\mathcal S(H, L) = {\tt min}_{H,L}\left(\bigcup_{W \in \mathcal S(H_v, L_v)} \mathcal S(H_{W} - v, L_{W}) + W \right) \cup {\tt min}_{H, L}\left(\mathcal S(H_v, L_v) + v \right)$.   
\end{lemma}
\begin{proof}
    By~\Cref{lem:large}, it suffices to show that $\mathcal S(H - v, L) = {\tt min}_{H,L}(\bigcup_{W \in \mathcal S(H_v, L_v)} \mathcal S(H_{W} - v, L_{W}) + W)$.

    Observe that each set in $\mathcal S(H_W - v, L_W) + W$ for $W \in \mathcal S(H_v, L_v)$ does not contain $v$.
    Moreover, such a set is a valid augmentation for $L$ in $H$ due to \Cref{lem:cvc:contraction-extension}.
    Thus, each set in ${\tt min}_{H, L}(\bigcup_{W \in \mathcal S(H_v, L_v)} \mathcal S(H_{W} - v, L_W) + W)$ is a minimal valid augmentation for $L$ in $H$ not containing $v$, meaning that
    \begin{align*}
        {\tt min}_{H, L}\left(\bigcup_{W \in \mathcal S(H_v, L_v)} \mathcal S(H_{W} - v, L_W) + W\right) \subseteq \mathcal S(H - v, L).
    \end{align*}

    Let $W' \in \mathcal S(H - v, L)$ be a minimal valid augmentation for $L$ in $H - v$.
    Since $H[L \cup W']$ is connected, $W'$ is also a valid augmentation for $L_v$ in $H_v$.
    This implies that there is $W \in \mathcal S(H_v, L_v)$ such that $W \subseteq W'$.
    By~\Cref{lem:cvc:contraction-minimality}, $W' \setminus W$ is a minimal valid augmentation for $L_W$ in $H_W$ and hence in $H_W - v$ as $v \notin W$.
    Hence, $W' \in \mathcal S(H_W - v, L_W) + W$.
    As $W' \in \mathcal S(H - v, L)$, every vertex in $W'$ is a cut vertex in $H[L \cup W']$.
    This ensures that $W' \in {\tt min}_{H,L}(\mathcal S(H_W - v, L_W) + W)$.
\end{proof}

The above lemmas enable us to enumerate minimal valid augmentations in $\mathcal S(H, L)$ with two methods.
Suppose that $L$ contains at least two vertices.
In our algorithm, we select a vertex $v \in R \coloneqq V(H) \setminus L$ of highest degree among $R$.
If $v$ has at most one neighbor, we conclude that there is no valid augmentation for $L$ in $H$ at all, as $|L| \ge 2$.
If $v$ has more than $|L| / 2$ neighbors in $L$,
we adopt the first method; we adopt the second method otherwise.
\Cref{alg:oq} recursively employs this branching strategy until $\min(|L|, |V(H) \setminus L|) \le 1$.

\begin{algorithm}[t]
    \SetAlgoLined
    \SetKwProg{Procedure}{Procedure}{}{}
    \SetKwFunction{EnumMinISSE}{MinValidAug}
    \DontPrintSemicolon
    \Procedure{\EnumMinISSE{$H, L$}}{
        Let $R = V(H) \setminus L$.\;
        \If{$\min(|L|, |R|) \le 1$}{
            \Return all minimal valid augmentations computed by a brute-force algorithm.
        }
        Let $v$ be a highest degree vertex in $R$.\;
        \If{$|N(v)| \le 1$}{\Return $\emptyset$}
        $\mathcal S \gets \emptyset$\;
        $\mathcal S_v \gets $\EnumMinISSE{$H_v, L_v$}\;
        \If{$|N(v)| > |L|/2$}{
            \Return \EnumMinISSE{$H - v, L$} $\cup~{\tt min}_{H, L}(\mathcal S_v + v)$\;
        }\Else{
            \ForEach{$W \in \mathcal S_v$}{
                $\mathcal S \gets \mathcal S \cup {\tt min}_{H, L}($\EnumMinISSE{$H_{W} - v, L_W$}$ + W)$\;
            }
            \Return $\mathcal S \cup {\tt min}_{H, L}(\mathcal S_v + v) $\;
        }
    }
    \caption{An output quasi-polynomial time algorithm for enumerating minimal valid augmentations. }
    \label{alg:oq}
\end{algorithm}

The correctness of \Cref{alg:oq} directly follows from \Cref{lem:large,lem:small}.
We then analyze the running time of \Cref{alg:oq}.
Let $T(\ell, r, N)$ be the worst case running time of \Cref{alg:oq} when a bipartite graph $H$ with a color class $L$ that satisfies $|L| \le \ell$, $|V(H)\setminus L| \le r$, and $|\mathcal S(H, L)| \le N$ is given as input.
Clearly, $T(\ell, r, N) \le T(\ell', r', N')$ for $\ell \le \ell'$, $r \le r'$, and $N \le N'$.

\begin{lemma}\label{lem:cvc:general:runnning-time}
    For $\ell, r, N \in \mathbb N$ with $n = \ell + r$, it holds that $T(\ell, r, N) = (n + N)^{O(\log n)}$.
\end{lemma}
\begin{proof}
    We prove the following claim by induction on $n$: $T(\ell, r, N) \le (\ell\cdot r + N)^{c\log (\ell\cdot r)}$ for some constant $c$.
    Let $H = (L \cup R, E)$ be a bipartite graph with $|L| \le \ell$, $|R| \le r$, and $|\mathcal S(H, L)| \le N$.
    If $\min(|L|, |V(H) \setminus L|) \le 1$, then $T(\ell, r, N)$ is upper bounded by a polynomial in $n$, as each set in $\mathcal S(H, L)$ has at most two vertices.
    Thus, we consider the other case.
    Observe that $|\mathcal S(H - v, L)| \le |\mathcal S(H, L)|$ for any $v \in R$.
    Moreover, we claim that $|\mathcal S(H_v, L_v)| \le |\mathcal S(H, L)|$ for any $v \in R$.

    Let $W \in \mathcal S(H_v, L_v)$.
    By~\Cref{lem:cvc:contraction-extension}, $W \cup \{v\}$ is a valid augmentation for $L$ in $H$.
    Moreover, all vertices in $W$ are cut vertices in $H[L \cup W]$ as they are cut vertices in $H_v[L_v \cup W]$.
    Thus, either $W$ or $W \cup \{v\}$ is a minimal valid augmentation for $L$ in $H$.
    This indicates that there is an injection $\phi$ from $\mathcal S(H_v, L_v)$ to $\mathcal S(H, L)$ by setting $\phi(W) = W$ or $\phi(W) = W \cup \{v\}$ for all $W \in \mathcal S(H_v, L_v)$, proving the claim $|\mathcal S(H_v, L_v)| \le |\mathcal S(H, L)|$.
    This also proves that $|\mathcal S(H_X, L_X)| \le |\mathcal S(H, L)|$ for $X \subseteq R$.

    Now, we turn to showing that $T(\ell, r, N) \le (\ell \cdot r + N)^{c\log (\ell \cdot r)}$ for some constant $c$.
    For simplicity, we assume that both $\ell$ and $r$ are even. 
    Observe that given a collection of valid augmentations $\mathcal A$, one can compute ${\tt min}_{H, L}(\mathcal A)$ in time $O(n^2|\mathcal A|)$ using a standard data structure.

    Suppose that there is a vertex $v \in R$ with $|N_H(v)| > |L|/2$.
    In this case, $L_v$ has at most $|L|/2$ vertices.
    Moreover, both $H_v$ and $H - v$ have at most $r - 1$ vertices on the ``right-hand side''.
    Since $|\mathcal S(H, L)| \ge |\mathcal S(H_v, L_v)|$, we have
    \begin{align*}
        T(\ell, r, N) \le \underbrace{T(\ell/2, r - 1, N)}_{\text{rec. call at line~9}} + \underbrace{T(\ell, r-1, N)}_{\text{rec. call at line~11}} + dn^2N
    \end{align*}
    for some constant $d$.

    Suppose otherwise that $|N_H(v)| \le |L|/2$.
    For $W \in \mathcal S(H_v, L_v)$, we have $L \setminus N_H(v) \subseteq N_H(W)$.
    As $H_v[L_v \cup W]$ is connected, each connected component of $H[L \cup W]$ contains at least one vertex in $N_H(v)$.
    This implies that $L_W$ contains at most $|N_H(v)| \le |L|/2$ vertices.
    Hence, we have
    \begin{align*}
        T(\ell, r, N) &\le \underbrace{T(\ell, r - 1, N)}_{\text{rec. call at line~9}} + \underbrace{|\mathcal S(H_v, L_v)| \cdot T(\ell/2, r-1, N)}_{\text{rec. calls at line~14}} + dn^2N \cdot |\mathcal S(H_v, L_v)|\\
        &\le T(\ell, r - 1, N) + N \cdot T(\ell/2, r-1, N) + dn^2N^2
    \end{align*}
    for some constant $d$.
    
    By combining the above two cases, we have
    \begin{align*}
        T(\ell, r, N) &\le T(\ell, r-1, N) + \max\{T(\ell/2, r - 1, N), N\cdot T(\ell/2, r-1, N)\} + \max\{dn^2N, dn^2N^2\}\\
        &= T(\ell, r-1, N) + N\cdot T(\ell/2, r-1, N) + dn^2N^2\\
        &\le T(\ell, r - 2, N) + N \cdot T(\ell/2, r - 2, N) + N \cdot T(\ell/2, r - 1, N) + 2dn^2N^2\\
        &\le T(\ell, r/2, N) + \frac{rN}{2}\cdot T(\ell, r - 1, N) + \frac{drn^2N^2}{2}\\
        &\le (\ell\cdot r/2 + N)^{c\log (\ell\cdot r/2)} + \frac{rN}{2}(\ell \cdot r/2 + N)^{c\log (\ell \cdot r/2)} + \frac{drn^2N^2}{2}\\
        &= \left(\frac{rN}{2}+1\right)(\ell\cdot r/2 + N)^{c\log (\ell\cdot r) - c} + \frac{drn^2N^2}{2}.
    \end{align*}
    By taking a sufficiently large constant $c$, we have $T(\ell, r, N) \le (\ell \cdot r + N)^{c\log(\ell\cdot r)}$.
\end{proof}

Overall, we have the following theorem.
\begin{theorem}\label{thm:cvc:general}
    There is an algorithm for enumerating minimal connected vertex covers of an $n$-vertex graph $G$ in total time $(n +  |\convc|)^{O(\log n)}$.
\end{theorem}

\section{\textsc{Minimal Connected Dominating Set Enumeration}}\label{sec:conds}
To enumerate all the minimal connected dominating sets, we use the same strategy in the previous section.
Let $G$ be a connected graph with maximum degree~$\Delta$ and let $\conds$ be the collection of all minimal connected dominating sets of $G$.
Let $X, Y \in \conds$ be distinct minimal connected dominating sets of $G$.
Then, there is a vertex $v \in X \setminus Y$.

\begin{lemma}\label{lem:conds:dominate-neighbors}
    There is a vertex set $W \subseteq Y$ of size at most $\Delta$ such that $(X \setminus \{v\}) \cup W$ is a dominating set of $G$.
    Moreover, $G[(X \setminus \{v\}) \cup W]$ has at most $\Delta$ components.
\end{lemma}
\begin{proof}
    Since $X$ is a dominating set of $G$, each vertex in $V(G) \setminus N[v]$ belongs to $X \setminus \{v\}$ or has a neighbor in $X \setminus \{v\}$.
    As $v \notin Y$, $Y$ contains at least one vertex~$v'$ in $N(v)$.
    Moreover, for each $w \in N(v) \setminus (X \cup \{v'\})$, $Y$ contains at least one vertex~$w'$ in $N[w]$.
    We let $W = \{v'\} \cup \{w' : w \in N(v) \setminus (X \cup \{v'\})\}$.
    Clearly, $W$ contains at most $|N(v) \setminus X| \le \Delta$ vertices.
    Each vertex in $N[v]$ belongs to $(X \setminus \{v\}) \cup W$ or has a neighbor in $(X \setminus \{v\}) \cup W$.
    This implies that $(X \setminus \{v\}) \cup W$ is a dominating set of $G$.

    Since $G[X]$ is connected, $G[(X \setminus \{v\})]$ has at most $|N(v) \cap X|$ components.
    As $|W| \le |N(v) \setminus X|$, $G[(X \setminus \{v\} \cup W]$ has at most $|N(v) \cap X| + |N(v) \setminus X| \le \Delta$ components.
\end{proof}

Let $W$ be a set of vertices of size at most $\Delta$ discussed in \Cref{lem:conds:dominate-neighbors} and let $X' = (X \setminus \{v\}) \cup W$ be a dominating set of $G$.
Since $Y$ is a dominating set of $G$, the following observation holds.

\begin{observation}\label{obs:conds:containment}
    For each component $C$ of $G[X']$, $N[C] \cap Y \neq \emptyset$.
\end{observation}

\begin{lemma}\label{lem:conds:augment}
    Let $S \subseteq V$ be an arbitrary dominating set of $G$ and let $Y$ be a minimal connected dominating set of $G$.
    Then, there is a vertex set $U \subseteq Y$ of at most $2q - 2$ vertices such that $G[S \cup U]$ is connected, where $q$ is the number of connected components in $G[S]$.
\end{lemma}
\begin{proof}
    We prove the lemma by induction on the number of connected components $q$ in $G[S]$.
    If $G[S]$ is connected, we are done.
    Suppose otherwise.
    Let $C$ and $C'$ be distinct components in $G[S]$.
    As both $N[C] \cap Y$ and $N[C'] \cap Y$ are nonempty (by \Cref{obs:conds:containment}), there is a path $P$ in $G[Y]$ connecting $N[C] \cap Y$ and $N[C'] \cap Y$.
    We choose $C$ and $C'$ to minimize the length of such a path $P$.
    The path $P$ contains at most two vertices of $V \setminus S$ as otherwise one of the internal vertices of $P$ has no neighbor in $S$, contradicting the fact that $S$ is a dominating set of $G$.
    Let $w,w'$ be these (possibly identical) two vertices of $V(P) \cap (V \setminus S)$. 
    Applying the induction to $S \cup \{w, w'\}$ proves the lemma.
\end{proof}

Thus, there is a set $U \subseteq Y$ of size at most $2\Delta - 2$ such that $G[X' \cup U]$ is a connected dominating set of $G$.
As $v \in X \setminus Y$ and $W \cup U \subseteq Y$, we have
\begin{align*}
    |X \setminus Y| > |(X \setminus \{v\}) \setminus Y|
    = |((X \setminus \{v\}) \cup W) \setminus Y|
    = |(X' \cup U)\setminus Y|
    \ge |Z \setminus Y|,
\end{align*}
where $Z$ is an arbitrary minimal connected dominating set of $G$ with $Z \subseteq X' \cup U$.
Now, let us formally define $\mathcal N^+(X)$ for $X \in \conds$.
For connected dominating set~$D$ of $G$, we let $\mu(D)$ be an arbitrary minimal connected dominating set of $G$ with $\mu(D) \subseteq D$.
By~\Cref{prop:monotone}, $\mu(D)$ can be computed in polynomial time.
We then define
\begin{align*}
    \mathcal N^+(X) &= \{\mu((X \setminus \{v\}) \cup W^*) : v \in X, W^* \subseteq V(G), |W^*| \le 3\Delta - 2,\\
    &\hspace{2cm} G[(X \setminus \{v\}) \cup W^*] \text{ is a connected dominating set of } G\}.
\end{align*}
By~\Cref{lem:conds:dominate-neighbors}, for each $Y \in \conds$,
there is a set of vertices $W \subseteq Y$ with cardinality at most $\Delta$ such that $(X \setminus \{v\}) \cup W$ is a dominating set of $G$. 
Moreover, by~\Cref{,lem:conds:augment}, there is a set of vertices $U \subseteq Y$ with cardinality at most $2\Delta - 2$ such that $(X \setminus \{v\}) \cup W \cup U$ is a connected dominating set of $G$.
Thus, for each $Y \in \conds$ with $X \neq Y$, $\mathcal N^+(X)$ has a minimal connected dominating set $Z \in \conds$ such that $|X \setminus Y| > |Z \setminus Y|$.
By~\Cref{thm:supergraph}, we have the following.

\begin{theorem}\label{thm:conds:degree}
    There is an $n^{3\Delta + O(1)}$-delay enumeration algorithm for \textsc{Minimal Connected Dominating Set Enumeration}.
\end{theorem}

We would like to mention that a similar improvement for $d$-claw free graphs would not be straightforward.
In fact, we can show that \textsc{Minimal Connected Dominating Set Enumeration} is at least as hard as \textsc{Minimal Transversal Enumeration}, even on claw-free graphs.
A graph is said to be \emph{co-bipartite} if its complement is bipartite.
It is easy to observe that every co-bipartite graph has no independent set of size $3$, and hence it is claw-free.
\begin{theorem}\label{thm:cds:mds-hard-cobipartite}
    If there is an output-polynomial time (or a polynomial-delay)  algorithm for \textsc{Minimal Connected Dominating Set Enumeration} on co-bipartite graphs, then there is an output-polynomial time (or a polynomial-delay) algorithm for enumerating minimal transversals in hypergraphs. 
\end{theorem}
\begin{proof}
    The reduction is in fact identical to one in \cite{KanteLMN:SIDMA:Enumeration:2014}, which proves a similar ``hardness'' of enumerating minimal dominating sets in co-bipartite graphs.

    Let $\mathcal H = (V, \mathcal E)$ be a hypergraph.
    We assume that $\mathcal E$ is nonempty and $\emptyset \not\in \mathcal E$ as otherwise the instance is trivial.
    From the incidence bipartite graph $\mathcal H$ (see~\Cref{thm:cvc:mds-hard}), we complete both $V$ and $V_{\mathcal E}$ into cliques.
    Then, we add a new vertex $r$, which is adjacent to all vertices in $V$.
    Let $G$ be the graph obtained in this way.
    It is easy to see that $G$ is co-bipartite and has two cliques $V \cup \{r\}$ and $V_{\mathcal E}$.

    Observe that every transversal $S \subseteq V $of $\mathcal H$ is a connected dominating set of~$G$.
    This follows from the facts that every vertex in $V_{\mathcal E}$ is adjacent to at least one vertex in $S$ and every vertex in $S$ is adjacent to other vertices in $V \cup \{r\}$.
    (Our assumption $\mathcal E\neq \emptyset$ implies that $S \neq \emptyset$.)

    Let $S$ be a minimal connected dominating set of $G$.
    As $r$ has no neighbors in $V_{\mathcal E}$, $S$ contains at least one vertex from $V \cup \{r\}$.
    Moreover, as $V_{\mathcal E} \neq \emptyset$ and $G[S]$ is connected, $S$ contains at least one vertex from $V$, implying that $r \notin S$ by the minimality of $S$.
    Then, observe that either $S = \{v, w_e\}$ for some $v \in V$ and $w_e \in V_{\mathcal E}$ with $v \in e$ or $S \subseteq V$.
    To see this, suppose that $S$ contains at least one vertex $w_e$ of $V_{\mathcal E}$.
    Then, as $S$ contains at least one vertex of $V$, it contains a vertex $v \in V$ that is adjacent to $w_e$ in $G$.
    Since these two vertices $v$ and $w_e$ dominate all the vertices in $G$, by the minimality of $S$, $S$ contains no other vertices at all.
    By a similar argument used in \Cref{thm:cvc:mds-hard}, $S$ is a minimal transversal of $\mathcal H$ if $S \subseteq V$ holds.

    Now, suppose that there is an output-polynomial time algorithm for \textsc{Minimal Connected Dominating Set Enumeration}.
    By applying the algorithm to $G$, we can enumerate all minimal connected dominating sets $\mathcal S$ of $G$.
    Since $\mathcal S$ contains all minimal transversals $\mathcal S'$ of $\mathcal H$ and $|\mathcal S \setminus \mathcal S'| = O(n^2)$, the algorithm enumerates all minimal transversals of $H$ in output-polynomial time as well.
\end{proof}

The above theorem reveals that the minimal connected dominating set enumeration is ``hard'' in some sense even for claw-free graphs, in contrast to the case of the minimal connected vertex cover enumeration~\Cref{thm:convc:claw}. 
Furthermore, the following theorem asserts that the minimal connected dominating set enumeration is also ``hard'' even for $2$-degenerate bipartite graphs. 
We would like to emphasize that the problem of enumerating minimal dominating sets in a bounded-degeneracy graph admits a polynomial-delay algorithm~\cite{EiterGM03}.
This indicates that the connectivity of solutions would make the problem more nontrivial, as seen in the case of \textsc{Minimal Connected Vertex Cover Enumeration}. 

\begin{theorem}
    If there is an output-polynomial time (or a polynomial-delay)  algorithm for \textsc{Minimal Connected Dominating Set Enumeration} on $2$-degenerate bipartite graphs, then there is an output-polynomial time (or a polynomial-delay) algorithm for enumerating minimal transversals in hypergraphs.     
\end{theorem}
\begin{proof}
    Our reduction is identical with \Cref{thm:cvc:mds-hard}. 
    Let $G'$ be the $2$-degenerate bipartite graph obtained from $\mathcal H$, which is constructed in the same way as in \Cref{thm:cvc:mds-hard}.
    We show that every minimal connected vertex cover of $G'$ is a minimal connected dominating sets of $G'$ and vice-versa.
    In the following discussion, we denote by $L$ the set of pendant vertices in $G'$ and
    by $U \coloneqq V(G') \setminus (V \cup L)$ the set of vertices in $G'$ having a pendant neighbor.

    Let $C$ be a minimal connected vertex cover of $G'$.
    As observed in~\Cref{thm:cvc:mds-hard}, we have $U \cap L = \emptyset$ and $U \subseteq C$.
    Since $U$ dominates all the vertices of $G$, $C$ is a connected dominating set of $G'$.
    To see the minimality, suppose that there is a vertex $u \in C$ such that $C \setminus \{u\}$ is still a connected dominating set of $G'$.
    Since each vertex in $U$ has a private pendant vertex (i.e., a pendant vertex that is only adjacent to it), $u \notin U$, which implies that $u \in V$.
    As $U$ covers all the edges of $G'$, $C \setminus \{u\}$ is also a minimal connected vertex cover of $G'$, deriving a contradiction.

    Let $D$ be a minimal connected dominating set of $G'$.
    Analogously, $D$ is a connected vertex cover of $G'$ containing $U$.
    Suppose that there is a vertex $u \in D$ such that $D \setminus \{u\}$ is a connected vertex cover of $G'$.
    As $U$ dominates all the vertices of $G'$ and $u \in D \setminus U$, $D \setminus \{u\}$ is a connected dominating set of $G'$, contradicting the minimality of $D$.
\end{proof}

\section{Capacitated vertex cover and dominating set}\label{sec:capacity}
This section is devoted to showing polynomial-delay algorithms for the capacitated variants on bounded-degree graphs.

The basic idea to prove \Cref{thm:convc:degree,thm:conds:degree} is that for distinct minimal solutions $X$ and $Y$ and for $v \in X \setminus Y$, there always exists a constant size (depending on $\Delta$) vertex set $W \subseteq Y$ such that $(X \setminus \{v\}) \cup W$ is a (not necessarily minimal) solution.
For \textsc{Minimal Capacitated Vertex Cover Enumeration} and \textsc{Minimal Capacitated Dominating Set Enumeration}, we can find such a set $W$ in bounded degree graphs as well.
Let $G$ be a graph with maximum degree~$\Delta$ and let $c\colon V(G) \to \mathbb N$.
Let $q = \max_{v \in V(G)}c(v)$.
Without loss of generality, we can assume that $q \le \Delta$.

\begin{lemma}\label{lem:capvc:augment}
    Let $X, Y$ be distinct minimal capacitated vertex covers of $(G, c)$ and let $v \in X \setminus Y$.
    Then, there is a vertex set $W \subseteq Y$ of size at most $q$ such that $(X \setminus \{v\}) \cup W$ is a capacitated vertex cover of $(G, c)$.
\end{lemma}
\begin{proof}
    Let $\alpha_X \colon E(G) \to X$ (resp. $\alpha_Y \colon E(G) \to Y$) be a function such that $(X, \alpha_X)$ (resp. $(Y, \alpha_Y)$) is a capacitated vertex cover of $(G, c)$.
    As $c(v) \le q$, $\alpha_X^{-1}(v)$ contains at most $q$ edges.
    For each $e = \{v, w\} \in \alpha_X^{-1}(v)$, we define a vertex $v_e \in Y$, and show that $(X \setminus \{v\}) \cup W$ with $W = \{v_e : e \in \alpha_X^{-1}(v)\}$ is a capacitated vertex cover of $(G, c)$.

    To define such a vertex $v_e$ for $e \in \alpha_X^{-1}(v)$, we use a similar strategy as used in \Cref{prop:feasibility-check}.
    Let $H$ be a bipartite graph with $V(H) = V_E \cup V$ such that $V_E = \{w_e : e \in E(G)\}$ and $V = \{w^i_u : u \in V(G), 1 \le i \le c(u)\}$.
    As shown in \Cref{prop:feasibility-check}, the capacitated vertex cover $(X, \alpha_X)$ (resp. $(Y, \alpha_Y)$) forms a matching $M_X$ (resp. $M_Y$) saturating $V_E$ in $H$ and vice versa.
    As $v \in X \setminus Y$, there is a vertex $w^i_v$ that is matched in $M_X$ but not matched in $M_Y$.
    Since both $M_X$ and $M_Y$ are maximum matchings in $H$, there is an alternating path $P$ between $w^i_v$ and $w^j_{u}$ for some $w^j_{u}$ that is not matched in $M_X$ but matched in $M_Y$.
    This implies that $u \in Y$, and we define $v_e \coloneqq u$.
    Then we obtain a new matching $(M_X \setminus E(P)) \cup (E(P) \cap M_Y)$ saturating $V_E$, which induces a capacitated vertex cover $(X\cup \{u\}, \alpha_{X \cup \{u\}})$ such that $|\alpha^{-1}_{X\cup \{u\}}(v)| < |\alpha^{-1}(v)|$.
    We repeat this for all other edges in $\alpha^{-1}_{X \cup \{u\}}(v)$, yielding a desired set $W$.
\end{proof}

An analogous lemma also holds for minimal capacitated dominating sets with a sightly more involved argument.
The proof is deferred to the full version due to the space limitation.

\begin{lemma}\label{lem:capds:augment}
    Let $X, Y$ be distinct minimal capacitated dominating sets of $(G, c)$ and let $v \in X \setminus Y$.
    Then, there is a vertex set $W \subseteq Y$ of size at most $q + 1$ such that $(X \setminus \{v\}) \cup W$ is a capacitated dominating set of $(G, c)$.
\end{lemma}
\begin{proof}
    Let $\beta_X \colon V(G) \setminus X \to X$ (resp. $\beta_Y \colon V(G) \setminus Y\to Y$) be a function such that $(X, \beta_X)$ (resp. $(Y, \beta_Y)$) is a capacitated dominating set of $(G, c)$.
    As $c(v) \le q$, $\beta_X^{-1}(v)$ contains at most $q$ vertices.
    For each $u \in \beta_X^{-1}(v)$, we define a vertex $u^* \in Y$, and show that $(X \setminus \{v\}) \cup W$ with $W = \{u^* : u \in \beta_X^{-1}(v) \cup \{v\}\}$ is a capacitated dominating set of $(G, c)$.

    To define such a vertex $u^*$ for $u \in \beta_X^{-1}(v) \cup \{v\}$, we again use a similar strategy as in \Cref{prop:feasibility-check}.
    Let $H$ be a bipartite graph with $V(H) = V' \cup V$ such that $V' = \{w_u : u \in V(G)\}$ and $V = \{w^i_u : u \in V(G), 1 \le i \le c(u)\}$.
    Abusing notation, for $A, B \subseteq V(G)$, we simply write $H[A, B]$ to denote the graph $H[\{w_u : u \in A\} \cup \{w_u^i: u \in B, 1 \le i \le c(u)\}]$.
    As shown in \Cref{prop:feasibility-check}, the capacitated dominating set $(X, \beta_X)$ (resp. $(Y, \beta_Y)$) forms a matching $M_X$ (resp. $M_Y$) saturating $\{w_u : u \in V(G) \setminus X\}$ (resp. $\{w_u : u \in V(G) \setminus Y\}$) in $H[V(G) \setminus X, X]$ (resp. $H[V(G)\setminus Y, Y]$) and vice versa.
    As $v \in X$, for $u \in \beta^{-1}_X(v)$, $M_X$ has an edge $\{w_u, w^{i^*}_v\}$ for some $i^*$.
    Moreover, as $v \notin Y$, $w^i_v$ has no incident edge in $M_Y$ for any $i$.
    Thus, the component of $M_X \triangle M_Y$ containing $w^{i^*}_v$ is an alternating path $P$ whose one end vertex is $w^{i^*}_v$.  
    \begin{figure}[t]
        \centering
        \includegraphics{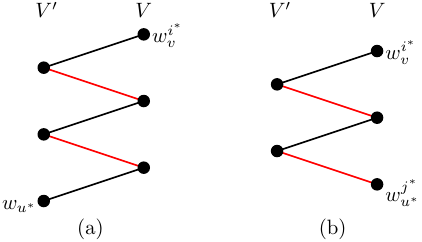}
        \caption{The figure depicts matchings $M_X$ (black edges) and $M_Y$ (red edges) in $H$.}
        \label{fig:capds-matching}
    \end{figure}
    
    Suppose that the other end vertex of $P$ belongs to $V'$.
    Let $w_{u^*}$ be the end vertex of $P$ other than $w^{i^*}_v$.
    In this case, we have $u^* \notin X$ as in~\Cref{fig:capds-matching}~(a).
    Then, $(M_X \setminus E(P)) \cup (M_Y \cap E(P))$ is a matching in $H[V(G)\setminus (X \cup \{u^*\}), X \cup \{u^*\}]$ saturating $\{w_x : x \in V(G) \setminus (X \cup \{u^*\})\}$.
    Suppose otherwise that the other end vertex of $P$ belongs to $V$.
    Let $w^{j^*}_{u^*}$ be the end vertex of $P$ other than $w^{i^*}_v$ as in~\Cref{fig:capds-matching}~(b).
    Then, $(M_X \setminus E(P)) \cup (M_Y \cap E(P))$ is a matching in $H[V(G)\setminus (X \cup \{u^*\}), X \cup \{u^*\}]$ saturating $\{w_x : x \in V(G)\setminus (X \cup \{u^*\})\}$.
    In both cases, we have $u^* \in Y$ as (a) $w_{u^*}$ has no incident edge in $M_Y$ or (b) $w^{j^*}_{u^*}$ has an incident edge in $M_Y$.
    Let $M^*$ be the matching in $H$ by applying one of the two cases (a) and (b) for all $u \in \beta^{-1}(v)$.
    Then, by letting $W^* = \{u^* \in Y : u \in \beta^{-1}(v)\}$, $M^*$ saturates $\{w_u : u \in V(G) \setminus ((X \setminus \{v\}) \cup W^*)\}$, and hence $(X \setminus \{v\}) \cup W^*$ dominates all the vertices of $V(G)$ except for $v$.

    Similarly, $w_v$ has an incident edge in $M_Y$ but no incident edge in $M^*$, there is an alternating path component $P$ in $M^* \triangle M_Y$ whose end vertex is $w_v$.
    Note that $P$ has at least one edge.
    Note moreover that $P$ has no vertices of $w^i_{v}$ for any $i$ as $M^*$ has no edges incident to them and $v \notin Y$.
    By applying the same argument as above, we can define $v^* \in Y$, and by letting $W = W^* \cup \{v^*\}$, there is a matching in $H[V(G) \setminus ((X \setminus \{v\}) \cup W), (X \setminus \{v\}) \cup W]$ saturating $\{w_u : u \in V(G) \setminus ((X \setminus \{v\}) \cup W)\}$.
    This implies that, by~\Cref{prop:feasibility-check}, $(X \setminus \{v\}) \cup W$ is a capacitated dominating set of $(G, c)$.
\end{proof}

The above two lemmas together with \Cref{prop:feasibility-check,thm:supergraph} yield polynomial-delay algorithms on bounded-degree graphs.
\begin{theorem}\label{thm:cap:alg}
    There are $n^{\min\{q, \Delta\} + O(1)}$-delay enumeration algorithms for \textsc{Minimal Capacitated Vertex Cover Enumeration} and \textsc{Minimal Capacitated Dominating Set Enumeration}, where $q = \max_{v \in V(G)}c(v)$ and $\Delta$ is the maximum degree of $G$.
\end{theorem}

Similarly to \Cref{thm:cvc:mds-hard}, \textsc{Minimal Capacitated Vertex Cover Enumeration} is ``not so easy'' even if an input graph is $2$-degenerate bipartite.

\begin{theorem}\label{thm:capvc:mds-hard}
    If there is an output-polynomial time (or a polynomial-delay) algorithm for \textsc{Minimal Capacitated Vertex Cover Enumeration} on $2$-degenerate bipartite graphs, then there is an output-polynomial time (or a polynomial-delay) algorithm for enumerating minimal transversals in hypergraphs.
\end{theorem}
\begin{proof}
    The underlying idea of the proof is analogous to \Cref{thm:cvc:mds-hard}.
    From a hypergraph $\mathcal H = (V, \mathcal E)$, we first construct a bipartite graph $G$ and a capacity function $c \colon V(G) \to \mathbb N$ and then transform them into a $2$-degenerate bipartite graph $G'$ and a function $c' \colon V(G') \to \mathbb N$.
    Notice that our reduction from the minimal transversal enumeration to 
    \textsc{Minimal Capacitated Vertex Cover Enumeration} is almost the same.

    We start with the incident bipartite graph of $\mathcal H$ (see~\Cref{thm:cvc:mds-hard} for details).
    For each $w_e \in V_{\mathcal E}$, we add a pendant vertex $w_e'$, and the obtained graph is denoted by $G$.
    See~\Cref{fig:capvc:reduction}~(a) for an illustration.
    \begin{figure}[t]
    \centering
    \includegraphics{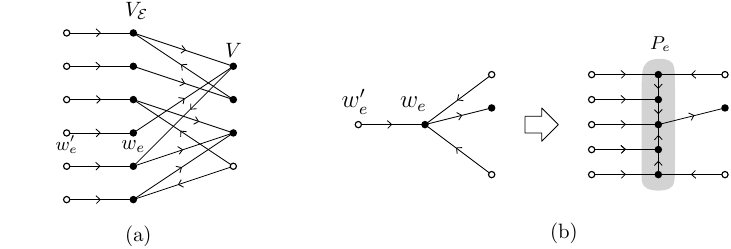}
    \caption{The left figure illustrates the constructed graph $G$. Filled circles indicate a minimal capacitated vertex cover of $(G, c)$ and arrows indicate the function $\alpha$.
    The right figure illustrates a part of the construction of a $2$-degenerate graph $G'$ from $G$. }
    \label{fig:capvc:reduction}
    \end{figure}
    We define the capacity function $c \colon V(G) \to \mathbb N$ as
    \begin{align*}
        c(v) = \begin{cases}
            d(v) & \text{if } v \in V\\
            d(v) - 1 & \text{if } v \in V_{\mathcal E}\\
            0 & \text{otherwise}.
        \end{cases}
    \end{align*}
    Here, $d(v)$ is the degree of $v$ in $G$.

    Let $C$ be an arbitrary minimal capacitated vertex cover of $(G, c)$ and let $\alpha \colon E(G) \to C$ be a function such that $(C, \alpha)$ is a capacitated vertex cover of $(G, c)$.
    As $c(w_e') = 0$ for $e \in \mathcal E$, we have $V_{\mathcal E} \subseteq C$.
    Moreover, $\alpha$ maps the incident edge of $w'_e$ to $w_e$ for each $e \in \mathcal E$.
    Since $c(w_e) = d(v) - 1$, $\alpha$ maps at least one edge incident to $w_e$ to a vertex $v$ in $C \cap V$.
    By the construction of $G$, such a vertex $v$ is contained in the hyperedge $e$.
    This implies that $S = C \cap V$ is a transversal of $\mathcal H$.
    The remaining part of this proof is analogous to that in \Cref{thm:cvc:mds-hard}.

    We now transform $(G, c)$ into $(G', c')$.
    The construction of $G'$ is again analogous to that in \Cref{thm:cvc:mds-hard}.
    For each $e \in \mathcal E$ with $|e| = d$, we replace $w_e$ and $w'_e$ with a path $P_{e} = (p^2_1, \ldots, p^e_{2d+1})$ and $2d+1$ pendant vertices in the same fashion as in \Cref{thm:cvc:mds-hard}.
    The obtained graph is denoted by $G'$.
    We define the capacity function $c' \colon V(G') \to \mathbb N$ as
    \begin{align*}
        c'(v) = \begin{cases}
            d'(v) & \text{if } v \in V\\
            d'(v) - 1 & \text{if } v \in V(P_e) \text{ for } e \in \mathcal E\\
            0 & \text{otherwise},
        \end{cases}
    \end{align*}
    where $d'(v)$ is the degree of $v$ in $G'$.
    (See \Cref{fig:capvc:reduction}~(b) for an illustration.)

    Let $C$ be an arbitrary minimal capacitated vertex cover of $(G', c')$ and let $\alpha \colon E(G') \to C$ be a function such that $(C, \alpha)$ is a capacitated vertex cover of $(G', c')$.
    As the capacity of the pendant vertex adjacent to each vertex $p^e_i$ in $V(P_e)$ is zero, $\alpha$ maps the edge between them to $p^e_i$, which implies that $V(P_e) \subseteq C$ for all $e \in \mathcal E$.
    Since the sum of the capacities of the vertices in $V(P_e)$ is exactly one less than the number of edges incident to $V(P_e)$, $\alpha$ maps at least one edge to a vertex in $V$.
    This allows us to simulate the correspondence between the family of minimal transversals of $\mathcal H$ and the family of minimal capacitated vertex covers of~$G$.
\end{proof}

\section{Concluding remarks}
In this paper, we present polynomial-delay algorithms for enumerating minimal connected/capacitated vertex covers/dominating sets on bounded-degree graphs.
Moreover, the algorithm for enumerating minimal connected vertex covers can be extended into the one working on $d$-claw free graphs.
In contrast to these positive results, we show that the problems of enumerating minimal connected/capacitated vertex covers in $2$-degenerate bipartite graphs are at least as hard as that of enumerating minimal transversals in hypergraphs, which has no known output-polynomial time algorithm until now.
For enumerating minimal connected vertex covers in general graphs, we develop an output quasi-polynomial time enumeration algorithm.

We would like to mention that we can enumerate minimal vertex covers satisfying both connectivity and capacity constraints in polynomial delay on bounded-degree graphs.
This can be done by just observing that for each minimal connected and capacitated vertex covers $X, Y$ of $(G, c)$ and for $v\in X$, there is a vertex set $W \subseteq Y$ of size at most $\Delta$ such that $(X \setminus v\}) \cup W$ is a capacitated vertex cover of $(G, c)$ (by \Cref{lem:capvc:augment}).
As $(X \setminus \{v\}) \cup W$ is a vertex cover of $G$, there is a vertex set $W' \subseteq Y$ of size at most $q-1$ such that $(X \setminus \{v\}) \cup W \cup W'$ is a connected vertex cover of $G$, where $q$ is the number of connected components in $G[(X \setminus \{v\}) \cup W]$ (by \Cref{lem:convc:augment}).
Since $G[(X \setminus \{v\}) \cup W]$ has at most $2\Delta$ components, the claim follows.
A similar argument also holds for the case of minimal connected and capacitated dominating sets in bounded-degree graphs.

\section*{Acknowledgements}
We are grateful to the anonymous reviewers for their careful reading of our manuscript and helpful comments.
\printbibliography

\end{document}